\RequirePackage{fix-cm}
\documentclass[twocolumn, draft, natbib, envcountsame]{svjour3}          
\smartqed  
\usepackage{graphicx}
\usepackage[misc]{ifsym}
\usepackage{todonotes}
\presetkeys{todonotes}{inline}{}
\usepackage{amsmath}
\usepackage{amssymb}
\usepackage{enumitem}
\usepackage{bbold}
\usepackage{mathtools}
\usepackage{hyperref}
\usepackage{algorithm}
\usepackage{algpseudocode}
\usepackage{tikz}
\usepackage{booktabs}
\usepackage{multirow}
\newcommand{\st}{\mathrel{}\middle|\mathrel{}}

\definecolor{mygray}{rgb}{0.9,0.9,0.9}

\newcommand{\grid}[4]{
	
	\draw [->] (0,-0.9*#3) -- (#1*#3 + 1.5*#3, -0.9*#3) node[anchor=#4] {time};
	
	\foreach \time in {0,...,#1}{
		\draw (\time*#3, -0.7*#3) -- (\time*#3, -1.1*#3) node[anchor=north] {\time};
	}
	\foreach \job in {1,...,#2}{
		\node[anchor=east] at (-0.3*#3, 0.5*#3 + #2*#3-\job*#3) {$J_{\job}$};
	}
}

\journalname{Journal of Scheduling}
\begin{document}
	
	\title{Online Algorithms to Schedule a Proportionate Flexible Flow Shop of Batching Machines
	}
	
	
	\author{Christoph Hertrich\textsuperscript{1,2}        \and
		Christian Wei\ss\textsuperscript{3}        \and
		Heiner Ackermann\textsuperscript{3}      \and
		Sandy Heydrich\textsuperscript{3}        \and
		Sven O. Krumke\textsuperscript{1}
	}
	
	\authorrunning{C. Hertrich, C. Wei\ss, H. Ackermann, S. Heydrich, S. O. Krumke} 
	
	\institute{
		\begin{itemize}
			\item[] Christoph Hertrich\\
			hertrich@math.tu-berlin.de
			\item[] Christian Wei\ss\\
			christian.weiss@itwm.fraunhofer.de
			\item[] Heiner Ackermann\\
			heiner.ackermann@itwm.fraunhofer.de
			\item[] Sandy Heydrich\\
			sandy.heydrich@itwm.fraunhofer.de
			\item[] Sven O. Krumke\\
			krumke@mathematik.uni-kl.de
		\end{itemize}
		\begin{itemize}
			\item[\textsuperscript{1}]
			Technische Universit\"at Kaiserslautern\\
			Department of Mathematics\\
			67663 Kaiserslautern, Germany
			\item[\textsuperscript{2}] 
			Technische Universit\"at Berlin\\
			Institute of Mathematics\\
			10623 Berlin, Germany
			\item[\textsuperscript{3}] 
			Fraunhofer Institute for Industrial Mathematics ITWM\\
			Department of Optimization\\
			67663 Kaiserslautern, Germany
		\end{itemize}
	}
	
	\date{Received: date / Accepted: date}
	
	\maketitle
	
	\begin{abstract}		
		This paper is the first to consider online algorithms to schedule a proportionate flexible flow shop of batching machines (PFFB). The scheduling model is motivated by manufacturing processes of individualized medicaments, which are used in modern medicine to treat some serious illnesses.
		We provide two different online algorithms, proving also lower bounds for the offline problem to compute their competitive ratios.
		The first algorithm is an easy-to-implement, general local scheduling heuristic. It is 2-competitive for PFFBs with an arbitrary number of stages and for several natural scheduling objectives. We also show that for total/average flow time, no deterministic algorithm with better competitive  ratio exists.
		For the special case with two stages and the makespan or total completion time objective, we describe an improved algorithm that achieves the best possible competitive ratio $\varphi=\frac{1+\sqrt{5}}{2}$, the golden ratio.
		All our results also hold for proportionate (non-flexible) flow shops of batching machines (PFB) for which this is also the first paper to study online algorithms.	
	
	\end{abstract}
	
	\section{Introduction}
	In modern pharmacy, in order to treat various serious illnesses, individualized medicaments are produced to order for a specific patient. These production processes often take place in a complex production line, consisting of many different steps.
	
	As long as each step of the process is still performed manually, for example by a laboratory worker, usually at each step only one patient can be handled at a time.
	However, once the process is scaled to industrial production levels, instead, for some steps, machines, like pipetting robots, are used.
	Indeed, at that point, in order to scale up production, usually several machines are used in parallel at each step.
	What is more, these types of machines can often handle multiple patients simultaneously. 
	If scheduled efficiently, this special feature can drastically increase the throughput of the production line.
	Clearly, in such an environment efficient operative planning  is crucial in order to optimize the performance of the manufacturing process and treat as many patients as possible as quickly as possible.
	
	In a practical setting, the producer of the individualized drug knows nothing about the patient until the medicine is actually ordered. This naturally creates an online scheduling scenario, for which efficient, computable and, if possible, easy-to-understand scheduling rules are needed. Therefore, in this paper, we specifically deal with the online version of the problem, although some of our findings are interesting in terms of offline scheduling as well.
		
	Formally, the manufacturing process studied in this paper is structured in a \emph{flexible flow shop} manner (also called \emph{hybrid flow shop} in the literature).
	A \emph{job} $J_j$, $j=1,2, \ldots, n$, representing the production of a drug for a specific patient, has to be processed in $s$ \emph{stages} $S_1,S_2, \ldots, S_s$ in order of their numbering. At each stage $S_i$ there are available $m_i$ identical, 
	parallel machines $M^{(1)}_i, M^{(2)}_i, \ldots, M^{(m_i)}_i$ to process the jobs. If each stage consists of only one machine, we may drop the machine index and instead identify each stage $S_i$ with its single machine $M_i$.
	
	Each job $J_j$ has a \emph{release date} $r_j \geq 0$, denoting the time at which the job $J_j$ is available for processing at the first stage $S_1$. We assume that jobs are indexed in earliest release date order.
	Furthermore, a job is only available for processing at stage $S_i$, $i=2,3, \ldots, s$, when it has finished processing at the previous stage $S_{i-1}$. 
	
	Processing times are only dependent on the stage, not on the job or the specific machine where the job is processed (recall that machines at each stage are identical). This means that each stage $S_i$, $i=1,2, \ldots, s$, is associated with a fixed \emph{processing time} $p_i$, which is the same for every job when processed at that stage on any machine.
	In the literature, a (flexible) flow shop with such job-independent processing times is sometimes called a \emph{proportionate} flow shop, see, e.g., \citet{Panwalkar:ProportionateReview}.
	
	Recall that, as a special feature from our application, each machine in the flexible flow shop can potentially handle multiple jobs at the same time. 
	These kind of machines are called \emph{(parallel) batching machines} and a set of jobs processed at the same time on some machine is called a \emph{batch} on that machine \citep[Chapter 8]{brucker:scheduling}.
	All jobs in one batch on some machine $M^{(k)}_i$ of stage $S_i$ have to start processing on $M^{(k)}_i$ at the same time.
	In particular, all jobs in one batch at stage $S_i$ have to be available for processing at $S_i$, before the batch can be started.
	The processing time of a batch on $M^{(k)}_i$ remains $p_i$, no matter how many jobs are included in this batch. 
	At each stage $S_i$, $i=1,2, \ldots, s$, machines have a common \emph{maximum batch size} (or \emph{batch capacity}) $b_i$, which is the maximum number of jobs a batch on machines of stage $S_i$ may contain.
	
	Given a feasible schedule $\varsigma$, we denote by $c_{ij}(\varsigma)$ the \emph{completion time} of job $J_j$ at stage $S_i$. For the completion time of job $J_j$ on the last machine we also write $C_j(\varsigma)=c_{sj}(\varsigma)$. If there is no confusion which schedule is considered, we may omit the reference to the schedule and simply write $c_{ij}$ and $C_j$.
	
	As optimization criteria, we study the four objective functions \emph{makespan} $C_{\max} = \max \{C_j \mid j= 1,2, \ldots, n \}$,
	\emph{total completion time} $\sum C_j = \sum_{j=1}^n C_j$, \emph{maximum flow time} $F_{\max} = \max \{F_j \mid j = 1, 2, \ldots, n \}$ and \emph{total flow time} $\sum F_j = \sum_{j=1}^n F_j$, where $F_j = C_j - r_j$. Note that the total flow time (or average flow time, if divided by the number of jobs) measures the average time a patient has to wait for his or her medicament, after the production is ordered. As short waiting times are essential in the treatment of life threatening illnesses, this objective is particularly relevant in practice. For more considerations about meaningful performance measures in applications, we refer to \citet{ackermannGOR}.
	
	Using the standard three-field notation for scheduling problems \citep{GrahamEtAl:ThreeFieldNot, PinedoScheduling},
	our problem is denoted as
	$$FFs \mid r_j, p_{ij} = p_i, p\text{-batch}, b_i \mid f,$$
	where $f$ is one of the four objective functions from above. We refer to the described scheduling model as \emph{proportionate flexible flow shop of batching machines} and abbreviate it by \emph{PFFB}. If we consider the special case where each stage consists of only one machine, we call this the usual \emph{proportionate flow shop of batching machines} and abbreviate it by \emph{PFB}.
	
	In this paper we deal with the \emph{online} problem to schedule a PFFB where each job is unknown until its release date. In particular, this means that the total number $n$ of jobs remains unknown until the end of the scheduling process.
	
	Throughout this paper, we write $\varphi=\frac{1+\sqrt{5}}{2}\approx1.618$ for the golden ratio.
	
	Next, we provide an (offline) example in order to illustrate the problem setting.
	
	\begin{example}\label{FirstExample}
		Consider an instance of PFFB with $s=2$ stages, $m_1=1$ machine at stage $S_1$ and $m_2=2$ machines at $S_2$, maximum batch sizes $b_1=3$ and $b_2=2$, processing times $p_1=3$ and $p_2=4$, as well as, $n=5$ jobs with release dates $r_1=r_2=0$, $r_3=1$, and $r_4=r_5=3$. Fig.~\ref{Fig:FirstExampleA} illustrates a feasible schedule for the instance as job-oriented Gantt chart.
		
		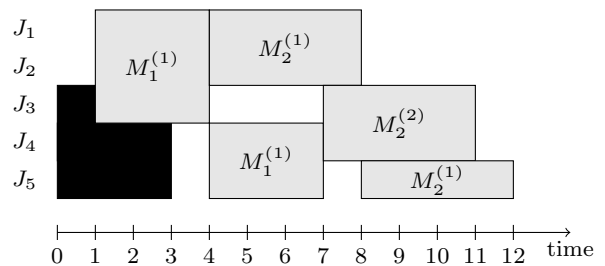
\begin{figure}[t]
			\centering
			\begin{tikzpicture}
			\newcommand{\step}{0.5}
			\grid{12}{5}{\step}{north}
			\draw [fill=black] (0, \step) rectangle (\step, 3*\step);
			\draw [fill=black] (0, 0) rectangle (3*\step, 2*\step);
			\draw [fill=mygray] (\step, 2*\step) rectangle (4*\step, 5*\step)  node [pos=0.5]{$M_1^{(1)}$};
			\draw [fill=mygray] (4*\step, 0) rectangle (7*\step, 2*\step) node [pos=0.5]{$M_1^{(1)}$};
			\draw [fill=mygray] (4*\step, 3*\step) rectangle (8*\step, 5*\step) node [pos=0.5]{$M_2^{(1)}$};
			\draw [fill=mygray] (7*\step, 1*\step) rectangle (11*\step, 3*\step) node [pos=0.5]{$M_2^{(2)}$};
			\draw [fill=mygray] (8*\step, 0) rectangle (12*\step, 1*\step) node [pos=0.5]{$M_2^{(1)}$};
			\end{tikzpicture}
			\caption{A feasible example schedule.}
			\label{Fig:FirstExampleA}
		\end{figure}
		
		Each rectangle labeled by a machine represents a batch of jobs processed together on this machine. The black area indicates that the respective jobs have not been released at this time yet. Note that in this example none of the batches can be started earlier, since either a job of the batch has just arrived when the batch is started, or the machine is occupied before. Still, the schedule does not minimize the makespan, since the schedule shown in Fig.~\ref{Fig:FirstExampleB} is feasible as well and has a makespan of 11 instead of 12.
		
		\begin{figure}[t]
			\centering
			\begin{tikzpicture}
			\newcommand{\step}{0.5}
			\grid{12}{5}{\step}{north}
			\draw [fill=black] (0, \step) rectangle (\step, 3*\step);
			\draw [fill=black] (0, 0) rectangle (3*\step, 2*\step);
			\draw [fill=mygray] (0, 3*\step) rectangle (3*\step, 5*\step)  node [pos=0.5]{$M_1^{(1)}$};
			\draw [fill=mygray] (3*\step, 0) rectangle (6*\step, 3*\step) node [pos=0.5]{$M_1^{(1)}$};
			\draw [fill=mygray] (3*\step, 3*\step) rectangle (7*\step, 5*\step) node [pos=0.5]{$M_2^{(1)}$};
			\draw [fill=mygray] (6*\step, 1*\step) rectangle (10*\step, 3*\step) node [pos=0.5]{$M_2^{(2)}$};
			\draw [fill=mygray] (7*\step, 0) rectangle (11*\step, 1*\step) node [pos=0.5]{$M_2^{(1)}$};
			\end{tikzpicture}
			\caption{An example schedule minimizing the makespan.}
			\label{Fig:FirstExampleB}
		\end{figure}
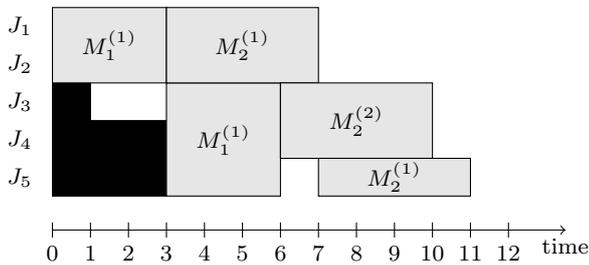
		
		The improvement in the makespan was achieved by reducing the size of the first batch on $M_1^{(1)}$ from three to two, which allows to start it one time step earlier. Observe that no job can arrive at $S_2$ before time step 3. Moreover, one of the two machines at $S_2$ has to process at least two batches in order to complete all five jobs. This takes at least 8 time units. Thus, 11 is the optimal makespan and no further improvement is possible.
	\end{example}
	
	\subsection{Our Results}
	
	This is the first study of P(F)FBs from an online perspective. We concentrate our research on the four practically highly relevant objective functions makespan, total completion time, maximum flow time and total flow time. An overview of our results can be found in Table~\ref{Tab:CompRatioTable}.
	
	\begin{table*} \centering%
		\begin{tabular}{cccccc}
			\toprule
			\shortstack{objective\\function} & \shortstack{number\\of stages} & \shortstack{lower\\bound} & \shortstack{upper\\bound} & \shortstack{reference\\lower bound} & \shortstack{reference\\upper bound}\\
			\midrule[\heavyrulewidth]
			\multirow{3}{*}{$C_{\max}$} & $1$ & $\varphi$ & $\varphi$ & \citet{zhang2003optimal} & \citet{zhang2003optimal} \\
			& $2$ & $\varphi$ & $\varphi$ & \citet{zhang2003optimal} & \textbf{Theorem \ref{Thm:TwoStages}}\\
			& $\geq3$ & $\varphi$ & $2$ & \citet{zhang2003optimal} & \textbf{Corollary \ref{Cor:NeverWait}}\\
			\midrule
			\multirow{2}{*}{$\sum C_j$} & $1,2$ & $\varphi$ & $\varphi$ & \textbf{Theorem \ref{Thm:sumCjLB}} & \textbf{Theorem \ref{Thm:TwoStages}}\\
			& $\geq3$ & $\varphi$ & $2$ & \textbf{Theorem \ref{Thm:sumCjLB}} & \textbf{Corollary \ref{Cor:NeverWait}}\\
			\midrule
			\multirow{2}{*}{$F_{\max}$} & $1$ & $\varphi$ & $\varphi$ & \citet{Jiao2014maxflowtime} & \citet{Jiao2014maxflowtime}\\
			 & $\geq2$ & $\varphi$ & $2$ & \citet{Jiao2014maxflowtime} & \textbf{Corollary \ref{Cor:NeverWait}}\\
			\midrule
			$\sum F_j$ & $\geq1$ & $2$ & $2$ & \textbf{Theorem \ref{Th:LBFSum}} & \textbf{Corollary \ref{Cor:NeverWait}}\\
			\bottomrule
		\end{tabular}%
		\caption{Known bounds for the competitive ratio of deterministic online algorithms to schedule a PFFB. Note that lower bounds for a single stage ($s=1$) carry over to arbitrary many stages by introducing stages with negligible processing times.}\label{Tab:CompRatioTable}%
	\end{table*}
	
	Concerning lower bounds on the competitive ratio, we observe that previously known bounds of $\varphi$ for one-stage PFFBs to minimize makespan or maximum flow time by \citet{zhang2003optimal} and \citet{Jiao2014maxflowtime}, respectively, carry over to arbitrarily many stages by introducing stages with negligible processing times. In addition, we show that $\varphi$ is also a lower bound with respect to total completion time, while for total flow time even a lower bound of $2$ can be achieved.
	
	Concerning upper bounds on the competitive ratio, we provide two algorithms. We first introduce the Never-Wait strategy, where machines are only idle if not enough patients are available for processing (see Section~\ref{Sec:NeverWait} for details). This strategy is easy to implement and applicable for all four considered objective functions and very general machine environments, even with large number of stages and/or large numbers of machines per stage. We prove that the Never-Wait strategy  achieves a competitive ratio of $2$ for any such instance and all four considered objectives. In particular, this is best-possible with respect to total flow time.
	We also show that the ``opposite'' strategy, where we always wait until a full batch can be started, is not an $\alpha$-approximation algorithm (and, in particular, not $\alpha$-competitive) for any $\alpha > 1$.
	For the specific scenario of $s = 2$ stages, we also introduce the $t$-Switch strategy, where the first stage is scheduled in such a way, that as many jobs as possible are available at the second stage at some time $t$; at the second stage, waiting is allowed until time $t$ and afterwards the stage is scheduled as in the Never-Wait strategy (see Section~\ref{Sec:PhiForTwo} for details). Choosing $t$ correctly yields an improved competitive ratio with respect to makespan and total completion time from $2$ to $\varphi$, which is also best possible for these problems.

	\subsection{Overview of this Paper}

	The remainder of this paper is structured as follows. In Section \ref{Sec:Lit}, we give an overview of related literature. In Section \ref{Sec:LowerBoundCmp} we provide lower bounds for the competitive ratio in PFFBs for the four objective functions we consider. Sections \ref{Sec:Permu} and \ref{Sec:LowerBound} are dedicated to establishing the structural results needed in order to prove competitiveness of the online algorithms we provide in the subsequent sections: in Section \ref{Sec:Permu} we prove that permutation schedules with jobs ordered by release dates are optimal for all of our four objective functions and in Section \ref{Sec:LowerBound} we show lower bounds for completion times of jobs in such permutation schedules. The latter are needed to prove competitiveness of our algorithms.
	Sections \ref{Sec:NeverWait} and \ref{Sec:PhiForTwo} deal with the general Never-Wait and the more specialized t-Switch algorithm, respectively.
	Finally, in Section \ref{Sec:Conc} we present conclusions and some further thoughts.

	\section{Literature}\label{Sec:Lit}

	To the best of our knowledge, neither the online PFFB problem, nor its special case, the online PFB problem, have been studied before.  A practical study specific to the application introduced in the beginning can be found in \cite{ackermannGOR}. Some of the concepts investigated in this paper in general have already been introduced in an application specific sense in that paper.
	
	Even for the offline PFFB problem, little is known. Previous work has been focused on (usually non-pro\-por\-tionate) NP-hard generalizations motivated by applications in various manufacturing industries. Common techniques include mixed-integer programming models or application-specific heuristics. See, e.g., \citet{aminnaseri2009hybridbatching, Luo2011hybridbatch, Li2015:HybridFlowShopBatchingGenetic, Tan2018:FF2IncompatibleWeightedTardiness}.
	
	In the special case of offline PFBs, with only one machine per stage, proportionate versions have been studied explicitly. \citet{SungEtAl:ProblemReduction} propose heuristic approaches to minimize the makespan and the total completion time in a PFB. However, they do not establish any complexity result.
	For the special case of two stages, \cite{Ahmadi:Batching}, as well as, \cite{SungYoon:DynamicArrivalsTwoBPM} present polynomial time algorithms. In a previous paper \citep{PFBPaper1}, we present a dynamic program that can be used to minimize several traditional objective functions (including the four objectives studied in this paper) in polynomial time for any fixed number $s$ of stages. For the case of $s$ being part of the input, the complexity status of offline PFBs is open. Significant hardness results have, to the best of our knowledge, not been achieved at all. See also \citep{PFBPaper1}, for an in-depth literature review for offline PFBs.
	
	Although the PFFB problem itself has not been investigated from an online or offline perspective before, there are several helpful results for related online problems in the literature. For our purposes, the most interesting family of related problems is online scheduling of single and parallel batching machines. With job-independent processing times, these problems can be viewed as the one-stage versions of P(F)FBs. We refer to \citet{Tian2014:OnlineSurvey} for a survey.
	
	Concerning online makespan minimization for a single batching machine with identical processing times, \citet{Zhang:OnlineBound} and \citet{Deng2003:OnlineBound} show that no deterministic online algorithm can achieve a competitive ratio better than the golden ratio $\varphi$. \citet{FangEtAl:OnlineGroupedProcTimes} provide a deterministic online algorithm matching this bound, even in a slightly more general setting, where processing times are not assumed to be identical, but only \emph{grouped}, that is, differing by a factor of at most $\varphi$ from each other. Moreover, \citet{zhang2003optimal} show that $\varphi$ is also the precise competitive ratio for makespan minimization on parallel batching machines with identical processing times, that is, a one-stage PFFB. \citet{li2018online} extend this result to minimizing the maximum weighted completion time.
	
	Concerning the total completion time objective, \citet{cao2011sumwjCj} present a 2-competitive online algorithm for parallel batching machines with identical processing times, which even works in the presence of precedence constraints. For the generalization where jobs are allowed to have unequal processing times and the total weighted completion time objective, $(4+\varepsilon)$-competitive algorithms are known \citep{chen2004line,ma2014sumwjCj}.
	
	Research on scheduling parallel batching machines to minimize maximum flow time started with the case where batches may have unbounded size \citep{li2011Fmax}. For the bounded batch model, as we consider it in this paper, \citet{Jiao2014maxflowtime} provide a $\varphi$-competitive deterministic online algorithm to minimize maximum flow time on parallel batching machines with identical processing times, i.e., a one-stage PFFB, and prove that this is best possible. Hence, the competitive ratio for makespan minimization by \citet{zhang2003optimal} carries over to maximum flow time. Recently, it has been shown that the competitive ratio of $\varphi$ for maximum flow time remains also valid for maximum weighted flow time \citep{chai2019online} or maximum flow time with delivery times \citep{lin2019delivery}.
	
	Although total flow time seems to be a very reasonable objective function from a practical perspective, we are not aware of any previous research about competitive algorithms for bounded p-batching problems on single or parallel machines to minimize total flow time.

	\section{Lower Bounds for the Competitive Ratio \label{Sec:LowerBoundCmp}}	
	In this section we provide lower bounds on the competitiveness of deterministic online algorithms for PFFB problems with our considered objective functions $C_{\max}$, $\sum C_j$, $F_{\max}$, and $\sum F_j$.
	
	Before we start, note that lower bounds on PFFBs with only one stage naturally extend to PFFBs with arbitrarily many stages by choosing negligible processing times for all stages but the first one. Therefore, in the following, we focus on such instances with $s=1$.
	
	For the makespan objective, \citet{zhang2003optimal} prove that the golden ratio $\varphi$ is a lower bound on parallel batching machines with identical processing times, i.e., a one-stage PFFB.	
	Furthermore, \cite{Jiao2014maxflowtime} prove the same bound of $\varphi$ for the maximum flow time objective.
	
	\begin{theorem}[\citealt{zhang2003optimal,Jiao2014maxflowtime}]\label{Thm:LBs}
		There are no deterministic online algorithms to minimize the makespan or maximum flow time in a PFFB with a competitive ratio less than $\varphi$. This result holds even if no stage has maximum batch size larger than $2$.
	\end{theorem}

	By using constructions similar to \citet{zhang2003optimal} and \citet{Jiao2014maxflowtime}, we also obtain lower bounds on the competitive ratio for the total completion and flow time objectives. We first show that the golden ratio $\varphi$ is also a lower bound for the competitiveness for the total completion time objective.
		
	\begin{theorem}\label{Thm:sumCjLB}
		There are no deterministic online algorithms to minimize the total completion time in a PFFB with a competitive ratio less than $\varphi$.
	\end{theorem}
	\begin{proof}
		Consider a PFFB instance with $s=1$, $m_1=1$, $p_1=1$ and a fixed batch size $b_1$. Suppose the first job $J_1$ is released at time $r_1=0$. We distinguish two possible behaviors of a deterministic online algorithm.
		
		\textbf{Case 1:} The algorithm starts a batch consisting of only $J_1$ at a point in time $t\leq\varphi-1$. In this case, suppose $b_1-1$ further jobs are released at time $t+\varepsilon$ for some small $\varepsilon>0$. Since these jobs can only be started after $J_1$ is finished, the schedule produced by the algorithm has a total completion time of at least $(t+1)+(b_1-1)(t+2)=b_1t+2b_1-1$. An optimal schedule would instead process all jobs in a single batch, resulting in a total completion time of $b_1(t+\varepsilon+1)$. If $\varepsilon$ tends to zero, this results in a competitive ratio of at least $\frac{b_1t+2b_1-1}{b_1(t+1)}\xrightarrow{b_1\rightarrow\infty}\frac{t+2}{t+1}\geq\frac{\varphi+1}{\varphi}=\varphi$.
		
		\textbf{Case 2:} The algorithm does not start a batch before time $\varphi-1$. In this case, suppose no further job is released and $J_1$ is the only job of the whole instance. The schedule produced by the algorithm has a total completion time of at least $\varphi$. An optimal schedule would start processing $J_1$ immediately at time $r_1=0$, resulting in a total completion time of $1$. Hence, also in this case, the competitive ratio is at least $\varphi$.		
		\qed
	\end{proof}
	
	Finally, for the total flow time, we can obtain a lower bound of $2$. The construction is very similar to the last proof, but instead of using $t= \varphi-1$
	as the border between the two cases, this time we use $t=1$. 
	
	\begin{theorem} \label{Th:LBFSum}
		There are no deterministic online algorithms to minimize the total flow time in a PFFB with competitive ratio less than $2$.
	\end{theorem}
	\begin{proof}
		As in the proof of Theorem~\ref{Thm:sumCjLB}, consider a PFFB instance with $s=1$, $m_1=1$, $p_1=1$ and a fixed batch size $b_1$. Again, we distinguish two cases, dependent on the start time of the first batch.
		
		\textbf{Case 1:} The algorithm starts a batch consisting of only $J_1$ at a point in time $t\leq1$. In this case, suppose $b_1-1$ further jobs are released at time $t+\varepsilon$ for some small $\varepsilon>0$. Since these jobs can only be started after $J_1$ is finished, the schedule produced by the algorithm has a total flow time of at least $(t+1)+(b_1-1)(t+2-t-\varepsilon)=t+1+(b_1-1)(2-\varepsilon)$. An optimal schedule would instead process all jobs in a single batch, resulting in a total flow time of $(t+\varepsilon+1)+b_1-1$. If $\varepsilon$ tends to zero, this results in a competitive ratio of at least $\frac{t+1+2(b_1-1)}{t+b_1}\xrightarrow{b_1\rightarrow\infty}2$.
		
		\textbf{Case 2:} The algorithm does not start a batch before time $1$. In this case, suppose no further job is released and $J_1$ is the only job of the whole instance. The schedule produced by the algorithm has a total flow time of at least $2$, while an optimal schedule could have started processing $J_1$ immediately at time $r_1=0$, resulting in a total flow time of $1$. Hence, also in this case, the competitive ratio is at least $2$.		
		\qed
	\end{proof}	

	\section{Optimality of Permutation Schedules \label{Sec:Permu}}
	
	In a \emph{permutation schedule} the order of the jobs is the same on all stages of the flexible flow shop. This means there exists a permutation $\pi$ of the job indices such that $c_{i\pi(1)}\leq c_{i\pi(2)} \leq \ldots \leq c_{i\pi(n)}$, for all $i=1,\ldots,s$. Since the processing times only depend on the stage and no preemption is possible, clearly the same then also holds for starting times instead of completion times.
	Note that this definition is not dependent on the specific machine where a job is scheduled. 
	If there exists an optimal schedule which is a permutation schedule with a certain ordering $\pi$ of the jobs, we say that \emph{permutation schedules are optimal}. A job ordering $\pi$ which gives rise to an optimal permutation schedule is then called an \emph{optimal job ordering}.
	Finally, an ordering $\pi$ of the jobs is called an \emph{earliest release date ordering}, if $r_{\pi(1)} \leq r_{\pi(2)} \leq \ldots \leq r_{\pi(n)}$.
	
	Using the techniques in the proofs of Lemma 2 and Theorem 3 from \citep{PFBPaper1}, we obtain the following theorem. For the interested reader, details of the proof are available in the appendix.
	
	\begin{theorem}\label{Thm:PermI}
		For a PFFB with objective function $C_{\max}$, $\sum C_j$, $F_{\max}$, or $\sum F_j$, permutation schedules are optimal. Moreover, any earliest release date order is an optimal ordering of the jobs.
	\end{theorem}

	From now on, we assume that jobs are indexed in earliest release date order and restrict our attention to permutation schedules where the job order is given by the indices. We therefore drop the notation of $\pi$. It remains to decide, for each stage, how the job set should be divided into batches on the individual machines and when to process these batches. In other words, every time a machine becomes idle and at least one job is available for processing, we have to decide how long to wait for the arrival of more jobs before starting the next batch. Waiting for additional jobs incurs the cost of delaying the already available jobs.
	
	\section{Lower Bounds for the Completion Times at each Stage}\label{Sec:LowerBound}
	When analyzing approximation or online algorithms, one typically compares the quality of the solution produced by the algorithm with the optimal offline solution. However, for many problems, including PFFBs, it is difficult to make a precise statement about the quality of such an optimal solution in the offline scenario. In these cases, a common approach is to use a lower bound for comparison instead. Therefore, in the first part of this section we develop a lower bound $c^*_{ij}$ for the completion time $c_{ij}$ of each job $J_j$ at each stage $S_i$. No feasible permutation schedule with job order given by the indices can yield smaller completion times. 
	
	In the second part of this section, we show how this lower bound can be interpreted as a solution of a proportionate flexible flow shop problem without batching machines.
	
	Finally, we conclude the section by comparing our bound to another one given by \citet{SungEtAl:ProblemReduction}.
	
	\subsection{Recursive formula for the lower bounds}
	
	We start by observing two properties that hold for the completion times of any permutation schedule with job order $J_1,J_2\dots,J_n$.
	
	Firstly, since a job needs to finish stage $S_{i-1}$ before it can be started at stage $S_i$, we have
	\begin{align}
	c_{ij} \geq c_{(i-1)j} + p_{i}\label{Equ:ProcTime}
	\end{align}
	for all $i\in\{2,3,\dots, s\}$ and $j\in[n]$.
	
	Secondly, for a stage index $i$ and a job index $j$, consider the two jobs $J_j$ and $J_{j-m_ib_i}$. Suppose there exists a point in time $t$ at which both jobs are simultaneously processed at stage $S_i$. Due to the fixed job permutation, this would imply that all the $m_ib_i+1$ jobs \mbox{$J_{j-m_ib_i}$, $J_{j-m_ib_i+1}$, $\dots$, $J_{j}$} would be processed at stage $S_i$ at this time $t$. However, this contradicts the batch capacity restriction $b_i$ of the $m_i$ machines. Hence, such a point in time $t$ cannot exist and we may conclude
	\begin{align}
	c_{ij}\geq c_{i(j-m_ib_i)} + p_i\label{Equ:BatchSize}
	\end{align}
	for all $i\in[s]$ and all $j=m_ib_i+1, m_ib_i+2,\dots,n$.
	
	Now we construct the desired lower bound via recursion, using the right-hand sides of the two properties above. 
	As starting values, we define
	$
	c^*_{0j}=r_j
	$
	for all $j\in[n]$ and
	$
	c^*_{ij}=-\infty
	$
	for $i\in[s]$ and $j\leq 0$. Then, we define
	\begin{align}
	c^*_{ij}=\max\{c^*_{(i-1)j}, c^*_{i(j-m_ib_i)}\}+p_i.\label{Equ:LowerBound}
	\end{align}
	
	Clearly, by inductive application of the two properties (\ref{Equ:ProcTime}) and (\ref{Equ:BatchSize}), the values $c^\ast_{ij}$ defined this way are a lower bound for the completion times of a permutation schedule for a PFFB, as stated in the following lemma.
	
	\begin{lemma}\label{Lem:LowerBound}
		Any feasible permutation schedule with job order $J_1,J_2,\dots,J_n$ satisfies $c_{ij}\geq c^*_{ij}$ for all $i\in[s]$ and $j\in[n]$.
	\end{lemma}

Before we proceed to the next part, we give an example to show that the lower bound cannot always be achieved by a feasible schedule. Hence, it is not tight in general.

\begin{example}\label{Exa:NotTight}
	Consider a PFFB instance without release dates, consisting of three stages with one machine each and two jobs (i.e., $s=3$, $m_1 = m_2 = m_3 = 1$, \mbox{$n=2$}). Processing times and batch capacities at each stage are given by $p_1=b_1=p_3=b_3=1$ and $p_2=b_2=2$. The only batching decision to make is whether to batch both jobs together at stage $S_2$ or not. The two corresponding permutation schedules are illustrated in Fig.~\ref{Fig:NotTight}. For ease of notation, we identify each stage $S_i$ with its single associated machine $M_i$. In both cases we have $c_{32}=6$. However, recursively applying (\ref{Equ:LowerBound}) yields $c^*_{32}=5$, illustrated in the third part of Fig.~\ref{Fig:NotTight}.
	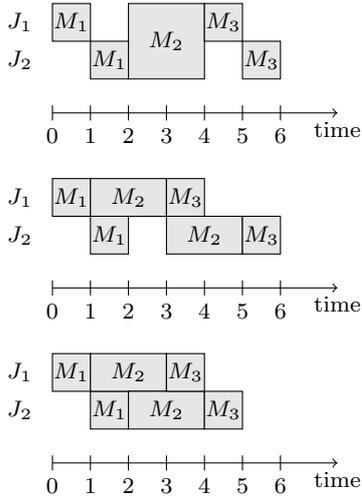
\begin{figure}[ht]
		\centering
		\begin{tikzpicture}
		\newcommand{\step}{0.5}
		\grid{6}{2}{\step}{north}
		\draw [fill=mygray] (0*\step, 1*\step) rectangle (1*\step, 2*\step) node [pos=0.5]{$M_1$};
		\draw [fill=mygray] (1*\step, 0*\step) rectangle (2*\step, 1*\step) node [pos=0.5]{$M_1$};
		\draw [fill=mygray] (2*\step, 0*\step) rectangle (4*\step, 2*\step) node [pos=0.5]{$M_2$};
		\draw [fill=mygray] (4*\step, 1*\step) rectangle (5*\step, 2*\step) node [pos=0.5]{$M_3$};
		\draw [fill=mygray] (5*\step, 0*\step) rectangle (6*\step, 1*\step) node [pos=0.5]{$M_3$};
		\end{tikzpicture}
		\vspace{1em}\\
		\begin{tikzpicture}
		\newcommand{\step}{0.5}
		\grid{6}{2}{\step}{north}
		\draw [fill=mygray] (0*\step, 1*\step) rectangle (1*\step, 2*\step) node [pos=0.5]{$M_1$};
		\draw [fill=mygray] (1*\step, 0*\step) rectangle (2*\step, 1*\step) node [pos=0.5]{$M_1$};
		\draw [fill=mygray] (1*\step, 1*\step) rectangle (3*\step, 2*\step) node [pos=0.5]{$M_2$};
		\draw [fill=mygray] (3*\step, 0*\step) rectangle (5*\step, 1*\step) node [pos=0.5]{$M_2$};
		\draw [fill=mygray] (3*\step, 1*\step) rectangle (4*\step, 2*\step) node [pos=0.5]{$M_3$};
		\draw [fill=mygray] (5*\step, 0*\step) rectangle (6*\step, 1*\step) node [pos=0.5]{$M_3$};
		\end{tikzpicture}
		\vspace{1em}\\
		\begin{tikzpicture}
		\newcommand{\step}{0.5}
		\grid{6}{2}{\step}{north}
		\draw [fill=mygray] (0*\step, 1*\step) rectangle (1*\step, 2*\step) node [pos=0.5]{$M_1$};
		\draw [fill=mygray] (1*\step, 0*\step) rectangle (2*\step, 1*\step) node [pos=0.5]{$M_1$};
		\draw [fill=mygray] (1*\step, 1*\step) rectangle (3*\step, 2*\step) node [pos=0.5]{$M_2$};
		\draw [fill=mygray] (2*\step, 0*\step) rectangle (4*\step, 1*\step) node [pos=0.5]{$M_2$};
		\draw [fill=mygray] (3*\step, 1*\step) rectangle (4*\step, 2*\step) node [pos=0.5]{$M_3$};
		\draw [fill=mygray] (4*\step, 0*\step) rectangle (5*\step, 1*\step) node [pos=0.5]{$M_3$};
		\end{tikzpicture}		
		\caption{The first two figures show the two optimal permutation schedules (ordered by their indices) in Example \ref{Exa:NotTight}; the third figure shows the \emph{infeasible} schedule implied by lower bound (\ref{Equ:LowerBound}), where machine $M_2$ is allowed to start a new batch while another is still running.}
		\label{Fig:NotTight}
	\end{figure}
\end{example}

\subsection{The lower bounds as completion times of a proportionate flexible flow shop}
Given an instance of PFFB, consider the following instance of a proportionate flexible flow shop (PFF) problem: at each stage, instead of $m_i$ parallel batching machines with maximum batch size $b_i$, there are $m_ib_i$ parallel machines with maximum batch size $1$. 
In other words, any batching machine with batch size $b_i$ at stage $S_i$ is replaced by $b_i$ identical parallel machines without batching. All other data of the instance remain the same.
Considering the original PFFB instance, we call the instance of PFF constructed this way \emph{the corresponding instance without batching}. 

The main difference between the PFFB instance and its corresponding instance without batching is that the $m_ib_i$ machines in the instance without batching can start jobs independently from one another, whereas in the PFFB instance, jobs in the same batch have to be started at the same time.
So the corresponding instance without batching allows for $m_ib_i$ independent starts, while the PFFB setting only allows for $m_i$ many.

Clearly, any feasible schedule for the PFFB instance implies a feasible schedule in the corresponding instance without batching, by keeping all start and finish times the same and only splitting up batches across machines such that one job runs on each machine.
The reverse does not work. Indeed, considering again Example \ref{Exa:NotTight}, we can see that the third part of Fig.~\ref{Fig:NotTight} shows a solution to the corresponding instance without batching which cannot be transformed into a solution for the PFFB instance.

The next theorem shows that an optimal schedule for the PFF instance is determined by the values $c^*_{ij}$, $i\in[s]$, $j\in[n]$, from (\ref{Equ:LowerBound}).

\begin{theorem}\label{Thm:PHF}
	Consider a PFFB instance with a regular objective function for which permutation schedules are optimal. Then the values $c^*_{ij}$, $i\in[s]$, $j\in[n]$, are the completion times of an optimal schedule of the corresponding PFF instance without batching.
\end{theorem}
\begin{proof}
Since there are no batching machines involved in a PFF, there is no need to wait for the arrival of other jobs in order to achieve a fuller batch. Therefore, an optimal permutation schedule for the PFF can be achieved by starting each job $J_j$ at each stage $S_i$ as soon as the following two conditions are satisfied:
\begin{itemize}
	\item The job $J_j$ has finished stage $S_{i-1}$ (or has been released, if $i=1$).
	\item There is a machine available for processing job $J_j$ at stage $S_i$ and all jobs $J_1, J_2, \ldots, J_{j-1}$ have already been started at that stage. Due to the fixed permutation and the number $m_ib_i$ of machines at stage $S_i$, this is the case as soon as job $J_{j-m_ib_i}$ has finished stage $S_i$ (or immediately, if $j \leq m_ib_i$).
\end{itemize}
Putting these conditions together, the completion time of job $J_j$ at stage $S_i$ in an optimal permutation schedule can be calculated recursively by
\begin{align*}
c_{ij}=\max\{c_{(i-1)j}, c_{i(j-m_ib_i)}\} + p_i.
\end{align*}
This is exactly the same formula as in the definition \eqref{Equ:LowerBound} of the values $c^*_{ij}$, $i\in[s]$, $j\in[n]$.\qed
\end{proof}

\subsection{Comparison with the bound of \citet{SungEtAl:ProblemReduction}}

As mentioned in the literature review, \citet{SungEtAl:ProblemReduction} propose heuristic algorithms to schedule a PFB, that is, a PFFB with only one machine at each stage, without release dates. In order to evaluate their experiments, they provide a lower bound for the makespan. Transferred to our notation it reads as follows:
\begin{align}
\begin{split}
\label{Equ:Sungbound}
\max\left\{\left\lceil\frac{j}{b_1}\right\rceil p_1 + \sum_{i=2}^{k-1} p_i + \left\lceil\frac{n-j+1}{b_k}\right\rceil p_k\right.\\\left.+ \sum_{i=k+1}^{s} p_i \st j\in[n], k\in[s] \right\}.
\end{split}
\end{align}

We show that our bound is an improvement, i.e., that $c^*_{sn}$ is at least as large as the bound given by (\ref{Equ:Sungbound}) and that for at least one instance, our bound is strictly larger.

Fix $j\in[n]$ and $k\in[s]$ as the maximizers in (\ref{Equ:Sungbound}). Applying the recursive definition (\ref{Equ:LowerBound}) with $m_i=1$ for all $i\in[s]$ and $r_j=0$ for all $j\in[n]$, we obtain the following four inequalities.
\begin{align*}
	c^*_{1j}&\geq\left\lceil\frac{j}{b_1}\right\rceil p_1,\\
	c^*_{kj}&\geq c^*_{1j}+\sum_{i=2}^{k} p_i,\\
	\begin{split}
	c^*_{kn}&\geq c^*_{kj}+\left\lfloor\frac{n-j}{b_k}\right\rfloor p_k\\
			&\qquad = c^*_{kj}+\left(\left\lceil\frac{n-j+1}{b_k}\right\rceil-1\right)p_k,
	\end{split}\\
	c^*_{sn}&\geq c^*_{kn} + \sum_{i=k+1}^{s} p_i.
\end{align*}
Summing up these four inequalities yields that $c^*_{sn}$ is at least as large as the bound (\ref{Equ:Sungbound}) by \citet{SungEtAl:ProblemReduction}. Moreover, the following example shows that $c^*_{sn}$ can be strictly larger than (\ref{Equ:Sungbound}). Consider a PFB (i.e., only one machine per stage) instance with $s=3$ stages. Let $n=6$ be the number of jobs and let $p_1=b_1=1$, $p_2=3$, $b_2=2$, $p_3=5$, and $b_3=3$. One can easily check that for this instance the maximum in $(\ref{Equ:Sungbound})$ is 16, which is attained either for $j=2$ and $k=2$ or for $j=3$ and $k=3$. However, for our bound it holds that
\begin{align*}
	c^*_{36}\geq 5+ c^*_{33}\geq 10 + c^*_{23}\geq 13+c^*_{21}\geq 16+c^*_{11}\geq 17.
\end{align*}
Hence, $c^*_{sn}$ is a strict improvement upon the bound $(4)$ of \citet{SungEtAl:ProblemReduction}.

We can compute $c^*_{sn}$ in time $\mathcal{O}(ns)$ by recursively using (\ref{Equ:LowerBound}). Since the computation of (\ref{Equ:Sungbound}) takes the same asymptotic runtime, our improvement of the bound incurs no increase in computational cost.

\section{The Never-Wait Algorithm}\label{Sec:NeverWait}
	
This section is devoted to establishing a simple, yet reasonably effective online scheduling rule for a PFFB. Again we focus on permutation schedules with job order $J_1, J_2,\dots,J_n$. Hence, on each stage, when a machine is idle and jobs are available, the main scheduling decision is how long to wait until starting a new batch. One obvious strategy is to always immediately start a batch when jobs are present and a machine is idle. We call this strategy the Never-Wait algorithm.
Another strategy, in some ways the opposite to the Never-Wait algorithm, is to always wait until a full batch can be started. This strategy is called the Full-Batch algorithm. 
In this section, we show that the Never-Wait algorithm is a 2-approximation with respect to various objective functions and, hence, 2-competitive when seen as an online algorithm. Furthermore, we show that the Full-Batch algorithm admits no constant approximation guarantee at all.
	
\begin{definition}
	The \emph{Never-Wait algorithm} for scheduling a PFFB is defined as follows: A batch is started at a stage whenever at least one job is available for processing at that stage and at least one machine is idle at that stage. The size of the new batch is chosen as large as possible, i.e., the minimum of the batch capacity at the stage and the number of available jobs at the stage when the batch is started.
	
	Note that, if at stage $S_i$ several machines are idle but there are not more than $b_i$ jobs available, only one machine is started. On the other hand, if there are more than $b_i$ jobs available at stage $S_i$, then more than one machine can be started at the same time.
\end{definition}

In the following, let $c_{ij}$, $i\in[s]$, $j\in[n]$, be the completion times resulting from the Never-Wait algorithm and $c^*_{ij}$ be the bound defined in Section \ref{Sec:LowerBound}. We also write $c_{0j}=c^*_{0j}=r_j$ for the the time at which $J_j$ becomes available at $S_1$.

Since the Never-Wait algorithm greedily starts as many available jobs as possible, the following property holds.

\begin{lemma} \label{lem:NeverWaitGreedy}
	Suppose that, for some $j \in \left[n \right]$ and some $i \in \left[ s \right]$, it holds that $c_{ij} > c_{(i-1)j} + 2 p_i$. 
	Then $j > m_ib_i$ and $c_{i(j-m_ib_i)} \geq c_{ij} - p_i$.
\end{lemma}
\begin{proof}
	Notice that $c_{ij} > c_{(i-1)j} + 2 p_i$ implies that $J_j$ is available but not started at stage $S_i$ during the complete interval $\lambda = \left[c_{ij}-2p_i,c_{ij}-p_i\right[$. Since interval $\lambda$ has length $p_i$ and at least one job is available during all of interval $\lambda$, each of the $m_i$ machines starts processing exactly one batch during $\lambda$. Moreover, since $J_j$ is already available, but not included in one of these batches, all these batches must be full batches. Hence, we obtain that at least $m_i b_i$ jobs complete stage $S_i$ in the time interval $\left[c_{ij}-p_i,c_{ij}\right[$. In particular, using the fixed job permutation, this implies that $j > m_ib_i$ and that $J_{j-m_ib_i}$ is completed not before time $c_{ij}-p_i$. \qed
\end{proof}

Now we are ready to show that the completion times produced by the Never-Wait algorithm can be bounded from above in terms of the lower bound $c_{ij}^*$ of the previous section.
	
\begin{theorem}\label{NeverWaitBound}
	For all $i\in[s]$ and $j\in[n]$, the completion time $c_{ij}$ in the Never-Wait algorithm satisfies
	\begin{align*}
	c_{ij}\leq c^*_{ij}+\sum_{i'=1}^i p_{i'}.
	\end{align*}
\end{theorem}
\begin{proof}
	We use a simultaneous induction on $i$ and $j$. Using $c_{0j}=r_j=c^*_{0j}$, the following arguments settle induction start ($i=1$) and induction step ($i>1$) at the same time.
	
	Let $i\in[s]$ be a stage index and let $j\in[n]$ be a job index. Suppose the claim is already proven for all pairs of indices $i'\leq i$ and $j'\leq j$ with either $i'<i$ or $j'<j$. We distinguish two cases.
	
	\textbf{Case 1}: $J_j$ waits for at most $p_i$ time units at stage $S_i$, i.e., $c_{ij}\leq 2p_i + c_{(i-1)j}$. In particular, by Lemma \ref{lem:NeverWaitGreedy}, this always holds if $j \leq m_ib_i$. We obtain
	\begin{align*}
	c_{ij}&\leq 2 p_i + c_{(i-1)j} \\&\stackrel{\text{ind.}}{\leq} p_i + p_i + c^*_{(i-1)j} + \sum_{i'=1}^{i-1} p_{i'} \\&\stackrel{\text{(\ref{Equ:LowerBound})}}{\leq} c^*_{ij}+\sum_{i'=1}^i p_{i'}.
	\end{align*}
	
	\textbf{Case 2}: $J_j$ waits for more than $p_i$ time units at stage $S_i$, i.e., $c_{ij}>c_{(i-1)j}+2p_i$. By Lemma \ref{lem:NeverWaitGreedy}, this implies that $j > m_ib_i$ and that $J_{j-m_ib_i}$ is completed not before time $c_{ij}-p_i$. Then we conclude
	\begin{align*}
	c_{ij}
	&\leq c_{i(j-m_ib_i)} + p_i \\
	&\stackrel{\text{ind.}}{\leq} c^*_{i(j-m_ib_i)}+p_i + \sum_{i'=1}^i p_{i'}\\
	&\stackrel{\text{(\ref{Equ:LowerBound})}}{\leq} c^*_{ij} + \sum_{i'=1}^i p_{i'}.\tag*{\qed}
	\end{align*}
\end{proof}
	
Using this, we obtain the desired competitiveness result.
	
\begin{corollary}\label{Cor:NeverWait}
	With respect to makespan, total completion time, maximum flow time and total flow time, the Never-Wait algorithm to schedule a PFFB is a $2$-com\-petitive online algorithm.
\end{corollary}
\begin{proof}
	Using Theorem \ref{Thm:PermI}, we obtain that there exists an optimal schedule $\varsigma^*$ that is a permutation schedule with job order $J_1,J_2,\dots,J_n$. Let $\varsigma$ be the schedule produced by the Never-Wait algorithm. Using Theorem \ref{NeverWaitBound} and Lemma \ref{Lem:LowerBound} it follows for a fixed job index $j\in[n]$ that
	\begin{align*}
	C_j(\varsigma) &\leq c^*_{sj} + \sum_{i=1}^s p_i \\&\leq C_j(\varsigma^*) + C_j(\varsigma^*) - r_j = 2C_j(\varsigma^*) - r_j,
	\end{align*}	
	where the term $- r_j$ stems from $C_j(\varsigma^\ast) \geq c_{sj}^\ast \geq r_j + \sum_{i=1}^s p_i$.
	Thus we have $C_j(\varsigma) \leq 2 C_j(\varsigma^\ast) - r_j \leq 2 C_j(\varsigma^\ast)$ and $F_j(\varsigma) = C_j(\varsigma) - r_j \leq 2C_j(\varsigma^\ast) - 2 r_j = 2 F_j(\varsigma^\ast)$, which proves the desired statement. \qed
\end{proof}

Corollary \ref{Cor:NeverWait} and Theorem \ref{Th:LBFSum} imply the following corollary.

\begin{corollary}
	For the total flow time objective, there is no deterministic online algorithm which, in general, has a better competitive ratio than the Never-Wait algorithm.
\end{corollary}

Next, we show by an example that the competitive ratio of the Never-Wait algorithm is not smaller than 2 with respect to any of the considered objectives.

\begin{example}\label{Exa:NeverWaitEpsilon}
	Consider the PFFB instance with only a single stage, with a fixed number $m_1\geq1$ of machines, a batch capacity $b_1 \geq 1$ to be chosen later, and  processing time $p_1=1$. For some small $\varepsilon>0$, suppose further there are $n=m_1b_1$ jobs, with release dates $r_j=(j-1)\varepsilon$ for $j\leq m_1$ and $r_j=m_1\varepsilon$ for $j\geq m_1+1$.
	
	The Never-Wait algorithm schedules the first $m_1$ jobs as singleton batches as soon as they arrive, filling all machines. All other jobs are not started before time $1$, when the first machine becomes idle again. Thus, none of the $n-m_1$ jobs $J_j$, $j \geq m_1 +1$, can be finished before time $2$. Hence, for the schedule $\varsigma$ produced by the \mbox{Never-Wait} algorithm, we obtain 
	\begin{align*}
	C_{\max}(\varsigma)&\geq 2,\\
	\sum C_j(\varsigma)&\geq 2n-m_1=2m_1b_1-m_1,\\
	F_{\max}(\varsigma)&\geq 2-m_1\varepsilon,\\
	\sum F_j(\varsigma)&\geq 2n-m_1-nm_1\varepsilon=2m_1b_1-m_1-m_1^2b_1\varepsilon.
	\end{align*}
	
	In contrast, consider the feasible schedule in which all machines remain idle until time $m_1\varepsilon$, when the last job becomes available. At this point in time, all the $n=m_1b_1$ jobs are partitioned into $m_1$ batches and started. For the resulting schedule $\varsigma'$ it follows that
	\begin{align*}
	C_{\max}(\varsigma')&\leq 1+m_1\varepsilon,\\
	\sum C_j(\varsigma')&\leq n + nm_1\varepsilon=m_1b_1+m_1^2b_1\varepsilon,\\
	F_{\max}(\varsigma')&\leq 1+m_1\varepsilon,\\
	\sum F_j(\varsigma')&\leq n + nm_1\varepsilon=m_1b_1+m_1^2b_1\varepsilon.
	\end{align*}
	
	Now, given any $\alpha> 0$, we show that $\varepsilon$ and $b_1$ can be chosen such that the ratio of the objective values is at least \mbox{$2-\alpha$}. Without loss of generality, assume that $\frac{4-\alpha}{\alpha}$ is an integer. If this is not the case, make $\alpha$ continuously smaller, until it is. Let $b_1 = \frac{4-\alpha}{\alpha}$ and $\varepsilon = \frac{1}{m_1 b_1}$.
	We obtain
	\begin{align*}
	\frac{C_{\max}(\varsigma)}{C_{\max}(\varsigma')}&\geq\frac{2}{1+m_1\varepsilon} = \frac{2}{1+ \frac{1}{b_1}} = \frac{2b_1}{b_1 + 1} \\ 
	& = \frac{8-2\alpha}{4} = 2 - \frac{1}{2} \alpha > 2 - \alpha;\\
	\frac{\sum C_j(\varsigma)}{\sum C_j(\varsigma')}&\geq\frac{2b_1-1}{b_1+m_1b_1\varepsilon} = \frac{2b_1-1}{b_1+ 1} = \frac{8-3 \alpha}{4} > 2 - \alpha; \\
	\frac{F_{\max}(\varsigma)}{F_{\max}(\varsigma')}&\geq\frac{2-m_1\varepsilon}{1+m_1\varepsilon} = \frac{2 - \frac{1}{b_1}}{1 + \frac{1}{b_1}} = \frac{2b_1 -1}{b_1 + 1} > 2 - \alpha; \\
	\frac{\sum F_j(\varsigma)}{\sum F_j(\varsigma')}&\geq\frac{2b_1-1-m_1b_1\varepsilon}{b_1+m_1b_1\varepsilon}
	 = \frac{2 b_1 - 2}{b_1+1} \\
	& = \frac{8 - 4 \alpha}{4} = 2 - \alpha.
	\end{align*}
		
	Hence, the Never-Wait algorithm does not have a competitive ratio less than 2 with respect to any of the considered objectives. Note that, while this example uses only a single stage, it can easily be extended to arbitrary many stages by using negligible processing times on all stages but the first one. Moreover, this example works for arbitrary values of $m_1$, which implies that the competitive ratio of the Never-Wait algorithm cannot be better than 2, no matter how many parallel machines per stage there are.
\end{example}

\begin{remark}\label{Rem:asymptotic1approx}
	Even though the Never-Wait algorithm is in general not better than 2-competitive, concerning the makespan and total completion time objectives, Theorem~\ref{NeverWaitBound} actually delivers a much stronger result than 2-competitiveness if the number of jobs is large.
	
	To see this, we first show that for any $i\in[s]$, we have $c^*_{ij}\geq\lceil \frac{j}{m_ib_i} \rceil p_i\eqqcolon\mathrm{LB}_{i,j}$. Indeed, with $k\coloneqq \lceil \frac{j}{m_ib_i} \rceil-1$ and $j'\coloneqq j-km_ib_i\geq 1$, recursive application of \eqref{Equ:LowerBound} yields
	\begin{align*}c^*_{ij}&\stackrel{\eqref{Equ:LowerBound}}{\geq} k p_i + c^*_{ij'} \textstyle\stackrel{\eqref{Equ:LowerBound}}{\geq} \lceil \frac{j}{m_ib_i} \rceil p_i + c^*_{(i-1)j'} \geq \mathrm{LB}_{i,j},\end{align*}
	where the last inequality follows from $c^*_{(i-1)j'}\geq 0$.
	
	Furthermore, $\mathrm{LB}_{i,j}$ tends to infinity for $j \rightarrow \infty$.
	Thus, in particular, if $n$ tends to infinity, also $c^*_{sn}$ tends to infinity, which is a lower bound for the makespan.
	
	In contrast, note that the difference between the makespan of the schedule produced by the Never-Wait algorithm and the optimal makespan is at most $\sum_{i=1}^s p_i$ by Theorem \ref{NeverWaitBound}. This difference stays constant if $n$ tends to infinity.
	
	Putting these things together, we obtain that, with respect to the makespan, the competitive ratio of the Never-Wait algorithm tends to 1, if $n$ tends to infinity.
	
	The same holds for the total completion time objective, as we argue now. Note that, with the same arguments as before, the total completion time of any feasible schedule is lower bounded by
	\[
	\sum_{j=1}^n c^*_{sj} \geq \sum_{j=1}^n \mathrm{LB}_{s,j} \geq \sum_{j=1}^n\left\lceil \frac{j}{m_sb_s} \right\rceil p_s \in \Omega(n^2).
	\]
	On the other hand, the difference between the total completion time of the schedule produced by the Never-Wait algorithm and the objective value of any optimal schedule is at most $n\sum_{i'=1}^s p_{i'}\in\mathcal{O}(n)$ by Theorem~\ref{NeverWaitBound}. Thus, also for the total completion time, we obtain that the competitive ratio of the Never-Wait algorithm tends to 1, if $n$ tends to infinity.
	
	However, a similar result cannot be achieved for our two flow time related objectives. To see this, note that Example~\ref{Exa:NeverWaitEpsilon} can be kind of ``copied'' arbitrarily often: given an instance with $n$ jobs for which the Never-Wait algorithm achieves a competitive ratio of at least $2-\alpha$ for some $\alpha>0$, introduce $n$ more jobs $J_{n+1}$ to $J_{2n}$ with release dates $r_j=r_{j-n}+M$, $j=n+1,n+2,\dots,2n$, for some large constant $M$. Then, both, the Never-Wait algorithm and an optimal offline algorithm, process the jobs $J_{n+1}$ to $J_{2n}$ in exactly the same way as they did process jobs $J_1$ to $J_n$, just shifted by $M$ time steps. By the definition of flow times (in contrast to completion times), this keeps the maximum flow time constant, while the total flow time is doubled for both algorithms. Thus, with respect to these two objectives, the competitive ratio is still at least $2-\alpha$. This procedure can be repeated arbitrarily often. Hence, the competitive ratio of the Never-Wait algorithm does not tend to 1 for $n\rightarrow\infty$ with respect to flow time related objectives.
\end{remark}

\subsection{The Full-Batch algorithm}

The Never-Wait algorithm can be seen as an extreme strategy, where all waiting time other than what is mandated by the scheduling constraints is avoided. Note that, for regular objective functions, it makes no sense to wait with starting a batch at stage $S_i$ if a machine is available and already $b_i$ jobs are waiting at stage $S_i$. Thus, it can be viewed as the opposite extreme to always wait until a full batch can be started.

\begin{definition}
	The \emph{Full-Batch algorithm} for scheduling a PFFB is defined by the following rule: At each stage, use only full batches (the last batch at stage $S_i$ may be less than full, if the batch capacity $b_i$ is not a divisor of the total number of jobs $n$).
\end{definition}

Note that the Full-Batch algorithm is not actually an online algorithm in the strict sense, because in order to know when to start the last batch one needs to know the number of jobs in advance. This is not the case in the standard online setting. Therefore, we view the Full-Batch algorithm primarily as an offline approximation algorithm. It can also be seen as an online algorithm in a relaxed online setting, where the scheduler receives the information that no more jobs are coming when the last job is released.

\begin{theorem}\label{Thm:FullBatch}
	There exists no $\alpha\geq1$ such that the Full-Batch algorithm is an $\alpha$-approximation for minimizing the makespan in a PFFB.
\end{theorem}
	\begin{proof}
		The statement already holds for usual PFBs, i.e. PFFBs where at each stage there is only one machine. Therefore, in the following, we assume $m_i=1$ for all stages $S_i$, $i\in[s]$. For ease of notation, we identify each stage $S_i$ with its single associated machine $M_i$.
		
		Without loss of generality let $\alpha\geq1$ be integer. Construct a PFB instance without release dates as follows. Set $s=10\alpha$, $n=5\alpha$, $p_{2i}=b_{2i}=1$, $p_{2i-1}=2$ and $b_{2i-1}=n$ for $i\in[5\alpha]$.
		
		Let $\varsigma$ be the schedule produced by the Full-Batch algorithm. Fig.~\ref{Fig:FullBatch} illustrates $\varsigma$ on the first four machines for the case $n=5$. An odd machine waits until all jobs have been finished on the previous even machine, before it processes all of them in a single batch. This way, we obtain $C_{\max}(\varsigma)=10\alpha + 5\alpha n = 10\alpha + 25\alpha^2$, where the first term stems from the $5\alpha$ odd machines and the second term from the $5\alpha$ even machines.
		
		\begin{figure}[ht]
			\centering
			\begin{tikzpicture}
			\newcommand{\step}{0.5}
			\grid{13}{5}{\step}{north}
			\draw [fill=mygray] (0*\step, 0*\step) rectangle (2*\step, 5*\step) node [pos=0.5]{$M_1$};
			\draw [fill=mygray] (2*\step, 4*\step) rectangle (3*\step, 5*\step) node [pos=0.5]{$M_2$};
			\draw [fill=mygray] (3*\step, 3*\step) rectangle (4*\step, 4*\step) node [pos=0.5]{$M_2$};
			\draw [fill=mygray] (4*\step, 2*\step) rectangle (5*\step, 3*\step) node [pos=0.5]{$M_2$};
			\draw [fill=mygray] (5*\step, 1*\step) rectangle (6*\step, 2*\step) node [pos=0.5]{$M_2$};
			\draw [fill=mygray] (6*\step, 0*\step) rectangle (7*\step, 1*\step) node [pos=0.5]{$M_2$};
			\draw [fill=mygray] (7*\step, 0*\step) rectangle (9*\step, 5*\step) node [pos=0.5]{$M_3$};
			\draw [fill=mygray] (9*\step, 4*\step) rectangle (10*\step, 5*\step) node [pos=0.5]{$M_4$};
			\draw [fill=mygray] (10*\step, 3*\step) rectangle (11*\step, 4*\step) node [pos=0.5]{$M_4$};
			\draw [fill=mygray] (11*\step, 2*\step) rectangle (12*\step, 3*\step) node [pos=0.5]{$M_4$};
			\draw [fill=mygray] (12*\step, 1*\step) rectangle (13*\step, 2*\step) node [pos=0.5]{$M_4$};
			\draw [fill=mygray] (13*\step, 0*\step) rectangle (14*\step, 1*\step) node [pos=0.5]{$M_4$};
			\node [] at (14.5*\step, 2.5*\step) {$\mathbf{\cdots}$};
			\end{tikzpicture}
			\caption{Schedule $\varsigma$ up to stage $S_4$ produced by the Full-Batch algorithm for the case $n=5$.}
			\label{Fig:FullBatch}
		\end{figure}
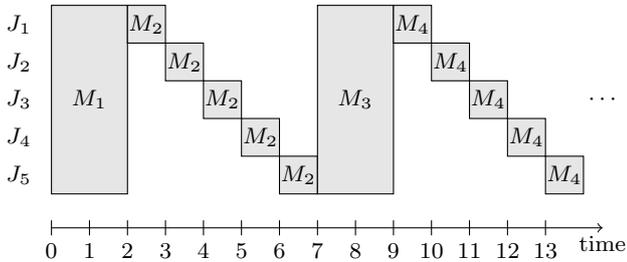
		
		In contrast, let $\varsigma'$ be the schedule where all batches consist of a single job only. On the first machine, each job is started two time steps after the previous job, i.e., $J_j$ is started at time $2j-2$. On all remaining machines, jobs can be started immediately upon arrival, since no processing time is larger than two. Fig.~\ref{Fig:OneByOne} illustrates $\varsigma'$ on the first four machines for the case $n=5$. We obtain \begin{align*}C_{\max}(\varsigma')=2n - 2 +15\alpha=25\alpha-2\end{align*} because it takes $2n-2$ time steps until the last job is started on the first machine, and $15\alpha$ more time steps for processing the last job on all machines.
		
		\begin{figure}[ht]
			\centering
			\begin{tikzpicture}
			\newcommand{\step}{0.5}
			\grid{13}{5}{\step}{north}
			
			
			\draw [fill=mygray] (0*\step, 4*\step) rectangle (2*\step, 5*\step) node [pos=0.5]{$M_1$};
			\draw [fill=mygray] (2*\step, 3*\step) rectangle (4*\step, 4*\step) node [pos=0.5]{$M_1$};
			\draw [fill=mygray] (4*\step, 2*\step) rectangle (6*\step, 3*\step) node [pos=0.5]{$M_1$};
			\draw [fill=mygray] (6*\step, 1*\step) rectangle (8*\step, 2*\step) node [pos=0.5]{$M_1$};
			\draw [fill=mygray] (8*\step, 0*\step) rectangle (10*\step, 1*\step) node [pos=0.5]{$M_1$};
			\draw [fill=mygray] (2*\step, 4*\step) rectangle (3*\step, 5*\step) node [pos=0.5]{$M_2$};
			\draw [fill=mygray] (4*\step, 3*\step) rectangle (5*\step, 4*\step) node [pos=0.5]{$M_2$};;
			\draw [fill=mygray] (6*\step, 2*\step) rectangle (7*\step, 3*\step) node [pos=0.5]{$M_2$};;
			\draw [fill=mygray] (8*\step, 1*\step) rectangle (9*\step, 2*\step) node [pos=0.5]{$M_2$};;
			\draw [fill=mygray] (10*\step, 0*\step) rectangle (11*\step, 1*\step) node [pos=0.5]{$M_2$};;
			\draw [fill=mygray] (3*\step, 4*\step) rectangle (5*\step, 5*\step) node [pos=0.5]{$M_3$};
			\draw [fill=mygray] (5*\step, 3*\step) rectangle (7*\step, 4*\step) node [pos=0.5]{$M_3$};
			\draw [fill=mygray] (7*\step, 2*\step) rectangle (9*\step, 3*\step) node [pos=0.5]{$M_3$};
			\draw [fill=mygray] (9*\step, 1*\step) rectangle (11*\step, 2*\step) node [pos=0.5]{$M_3$};
			\draw [fill=mygray] (11*\step, 0*\step) rectangle (13*\step, 1*\step) node [pos=0.5]{$M_3$};
			\draw [fill=mygray] (5*\step, 4*\step) rectangle (6*\step, 5*\step) node [pos=0.5]{$M_4$};
			\draw [fill=mygray] (7*\step, 3*\step) rectangle (8*\step, 4*\step) node [pos=0.5]{$M_4$};
			\draw [fill=mygray] (9*\step, 2*\step) rectangle (10*\step, 3*\step) node [pos=0.5]{$M_4$};
			\draw [fill=mygray] (11*\step, 1*\step) rectangle (12*\step, 2*\step) node [pos=0.5]{$M_4$};
			\draw [fill=mygray] (13*\step, 0*\step) rectangle (14*\step, 1*\step) node [pos=0.5]{$M_4$};
			\node [] at (14.5*\step, 2.5*\step) {$\mathbf{\cdots}$};
			\end{tikzpicture}
			\caption{Schedule $\varsigma'$ up to stage $S_4$ produced by using only batches of size one for the case $n=5$.}
			\label{Fig:OneByOne}
		\end{figure}
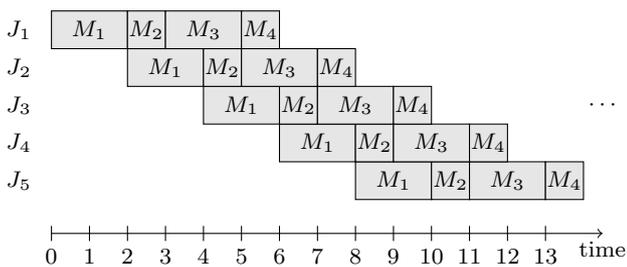
		
		Hence, in total, we obtain
		\begin{align*}
		\frac{C_{\max}(\varsigma)}{C_{\max}(\varsigma')}=
		\frac{10\alpha + 25\alpha^2}{25\alpha-2}>
		\frac{25\alpha^2}{25\alpha}=\alpha.
		\end{align*}
		
		Therefore, the Full-Batch algorithm cannot be an $\alpha$-approximation for any constant $\alpha\geq1$.\qed
	\end{proof}

	\begin{remark}
		Note that \citet{Hertrich2018} shows for the case of PFBs that the example in the proof of Theorem \ref{Thm:FullBatch} is no longer valid if either the number of jobs $n$ or the number of stages $s$ is fixed and no longer depends on $\alpha$. In these cases, the Full-Batch algorithm becomes a constant factor approximation. Analyzing the proofs of \citet{Hertrich2018}, one can see that these results carry over to PFFBs.
	\end{remark}

	\section{Optimal Online Algorithm for Two Stages}\label{Sec:PhiForTwo}
	
	We have seen that the competitive ratio of the Never-Wait algorithm is $2$ with respect to all four objective functions considered in this paper. Comparing with the lower bounds of Section~\ref{Sec:LowerBoundCmp}, this is best possible for the total flow time objective. However, for the other three objectives, there is a gap between the lower bound of the golden ratio $\varphi$ and the upper bound of $2$. In this section, we close this gap in the special case of $s\leq 2$ for makespan and total completion time by presenting a specialized $\varphi$-competitive algorithm for this case. This extends the result of \citet{zhang2003optimal}, who provide a $\varphi$-competitive algorithm for makespan minimization with $s=1$, i.e., on identical, parallel batching machines.
	
	Let $t=\varphi p_1+ (\varphi-1)p_2$. This is the latest possible time at which the first batch must be started at the second stage if we want that a job released at time zero is completed at time $\varphi(p_1+p_2)$, which is the minimal completion time of such a job multiplied with $\varphi$. The idea of the following algorithm is to schedule $S_1$ in a way such that as many jobs as possible have completed $S_1$ at time $t$, while the machines of $S_2$ stay idle until time $t$ and are scheduled according to the Never-Wait algorithm afterwards.
	
	\begin{definition}
		For a 2-stage PFFB the \emph{$t$-Switch} algorithm is defined as follows. Let the set $I$ of \emph{starting instants} consist of those points in time $\tau \geq 0$, for which $\tau + \ell p_1 = t$ for some integer $\ell \in \mathbb{Z}$, i.e.,		
		\begin{align*}
		I = \biggl\{ t-\left\lfloor\frac{t}{p_1}\right\rfloor p_1,\ & t-\left(\left\lfloor\frac{t}{p_1}\right\rfloor-1\right) p_1,\dots, \\
		&  t-p_1, t, t+p_1,\dots\biggr\}.
		\end{align*}
		At stage $S_1$, jobs are started only at starting instants. At each starting instant, as many jobs are started as possible, i.e., the minimum of $m_1b_1$ and the number of available jobs. The machines of $S_2$ stay idle until time $t$ and are scheduled according to the Never-Wait algorithm afterwards.
	\end{definition}
	
	Next, we prove three lemmas that help to show \mbox{$\varphi$-competitiveness}. In the following, let $c_{1j}$ and $c_{2j}$, $j\in[n]$, be the completion times produced by the $t$-Switch algorithm and let $c^*_{1j}$, $c^*_{2j}$, $j\in[n]$, be the lower bounds of Section \ref{Sec:LowerBound}.
	
	\begin{lemma}\label{Lem:FirstMachine}
		For all $j\in[n]$, it holds that $c_{1j}\leq c^*_{1j} + p_1$.
	\end{lemma}
		\begin{proof}
			We use induction on $j$.
			Fix a job index $j\in[n]$. Note that $J_j$ is started at the starting instant $\tau= c_{1j}-p_1$. Let $\tau^\prime=c_{1j}-2p_1$ be the previous starting instant. First, suppose that $r_j > \tau^\prime=c_{1j}-2p_1$, i.e., $\tau$ is the first starting instant after $r_j$. Note, in particular, that since $m_1b_1$ jobs can be started at each starting instant, this is always the case if $j \leq m_1b_1$. 
			In this case, we directly obtain $c_{1j}\leq r_j+2p_1\leq c^*_{1j}+p_1$.
						
			Otherwise, if $r_j\leq \tau^\prime=c_{1j}-2p_1$, then $J_j$ was already released at starting instant $\tau^\prime$, but has not been started at this time. Hence, $j > m_1b_1$ and exactly $m_1b_1$ other jobs must have been started at $\tau^\prime$. This implies $c_{1(j-m_1b_1)}\geq c_{1j}-p_1$ and, hence, 
			\begin{align*}
				c_{1j}&\leq c_{1(j-m_1b_1)}+p_1
				\\&\stackrel{\text{ind.}}{\leq} c^*_{1(j-m_1b_1)} + 2p_1
				\\&\stackrel{\text{(\ref{Equ:LowerBound})}}{\leq} c^*_{1j} + p_1 \tag*{\qed}
			\end{align*}
		\end{proof}
	
	\begin{lemma}\label{Lem:PhiEarly}
		For all $j\in[n]$ with $c_{1j}\leq t$, it holds that $c_{2j} \leq c^*_{2j} + (\varphi - 1)(p_1 + p_2)$.
	\end{lemma}
		\begin{proof}
			We use induction on $j$. Suppose first that $J_j$ is started at stage $S_2$ exactly at time $t$. Note that this is always the case if $j \leq m_2b_2$, as stage $S_2$ is idle before time $t$ and thus up to $m_2b_2$ jobs can be started at time $t$ (recall that it is assumed that $J_j$ is finished at $S_1$ before time $t$). We obtain 
			\begin{align*}
			c_{2j}&=t+p_2 = \varphi(p_1+p_2)\\&=(p_1+p_2) + (\varphi - 1)(p_1+p_2)\\&\leq c^*_{2j} + (\varphi - 1)(p_1+p_2).
			\end{align*}
			
			On the other hand, consider the case in which $J_j$ is started at stage $S_2$ later than time $t$. This can only happen if all machines of $S_2$ continuously process full batches between time $t$ and time $c_{2j}-p_2$. In particular, this implies $j > m_2b_2$ and $c_{2(j-m_2b_2)}\geq c_{2j}-p_2$. Hence,
			\begin{align*}
			c_{2j}&\leq c_{2(j-m_2b_2)}+p_2 \\&\stackrel{\text{ind.}}{\leq} c^*_{2(j-m_2b_2)}+p_2 + (\varphi - 1)(p_1 + p_2)\\&\stackrel{\text{(\ref{Equ:LowerBound})}}{\leq} c^*_{2j} + (\varphi - 1)(p_1+p_2).\tag*{\qed}
			\end{align*}
		\end{proof}
	
	\begin{lemma}\label{Lem:PhiLate}
		For all $j\in[n]$, it holds $c_{2j}< c^*_{2j} + p_1 + p_2$.
	\end{lemma}
		\begin{proof}
			If $c_{1j}\leq t$, then the claim for this index $j$ follows by Lemma \ref{Lem:PhiEarly}. If $c_{1j}> t$, then $S_2$ is already scheduled according to the Never-Wait algorithm when $J_j$ arrives. Hence, the claim can be proven analogously to Theorem~\ref{NeverWaitBound}, making use of Lemma \ref{Lem:FirstMachine}.\qed
		\end{proof}
	
	Now we are ready to prove $\varphi$-competitiveness of the $t$-Switch algorithm.
	
	\begin{theorem}\label{Thm:TwoStages}
		For a two-stage PFFB, the $t$-Switch algorithm is $\varphi$-competitive with respect to the two objective functions $C_{\max}$ and $\sum C_j$.
	\end{theorem}
		\begin{proof}
			Consider a job $J_j$, $j\in[n]$. We distinguish two cases:
			
			\textbf{Case 1:} $c^*_{1j}<t$. Using Lemma \ref{Lem:FirstMachine}, it follows that $c_{1j}<t+p_1$. Since $c_{1j}$ must be a starting instant, we even obtain $c_{1j}\leq t$. Now Lemma \ref{Lem:PhiEarly} yields
			\begin{align*}
			c_{2j} \leq c^*_{2j} + (\varphi - 1)(p_1 + p_2) \leq c^*_{2j} + (\varphi - 1)c^*_{2j} = \varphi c^*_{2j}.
			\end{align*}
			
			\textbf{Case 2:} $c^*_{1j}\geq t$. Then it follows that $c^*_{2j}\geq t+p_2 = \varphi(p_1+p_2)$. Using Lemma \ref{Lem:PhiLate}, we obtain
			\begin{align*}
			c_{2j} \leq c^*_{2j} + p_1 + p_2 \leq c^*_{2j} + \frac{1}{\varphi}c^*_{2j} = \varphi c^*_{2j}.
			\end{align*}
			Having proven $c_{2j} \leq \varphi c^*_{2j}$ for all $j\in[n]$, the $\varphi$-com\-pet\-i\-tive\-ness follows for $C_{\max}$ and $\sum C_j$.\qed
		\end{proof}
	
	Theorem \ref{Thm:TwoStages} in combination with Theorems~\ref{Thm:LBs} and~\ref{Thm:sumCjLB} implies the following corollary.
	
	\begin{corollary}
		For PFFBs with $s=2$ stages and the makespan or total completion time objective, there is no deterministic online algorithm which, in general, has a better competitive ratio than the $t$-Switch algorithm.
	\end{corollary}
	
	\section{Conclusion} \label{Sec:Conc}
	
	In this paper, we consider proportionate flexible flow shops with batching machines (PFFBs). We put a special focus on the online version of the problem, which is highly relevant for applications in the production of modern, individualized medicaments.
	To the best of our knowledge, the online version has not been studied before, not even in the special case of proportionate (non-flexible) flow shops with batching machines (PFBs).
	
	We describe and analyze two algorithms: the very general Never-Wait algorithm and the more specialized t-Switch algorithm. 
	The Never-Wait algorithm works for an arbitrary number of stages and machines. What is more, its description is relatively simple and therefore it is easy to implement in practice.
	We show that, despite its simplicity, the Never-Wait algorithm is 2-competitive for minimizing the makespan, total completion time, maximum flow time, and total flow time.
	Furthermore, we show that for the total flow time objective, no deterministic online algorithm can, in general, do better than the Never-Wait algorithm. 
	Note that the total flow time, which is equivalent with the average flow time by dividing by the constant number of jobs, is particularly important for our application: it measures the average time patients have to wait for their medicament after production is ordered. Obviously, a low average waiting time is necessary for patients to benefit from the medicine as quickly as possible.
	Interestingly, studying a particular industrial instance, \cite{ackermannGOR} have also done some initial work to confirm that the theoretical usefulness of the Never-Wait algorithm is also coherent with its practical behaviour.
	
	The t-Switch algorithm is specialized for PFFBs with only two stages and the makespan or total completion time objective. For these versions of the online problem, the t-Switch algorithm is a $\varphi$-competitive algorithm, with 
	$\varphi = \frac{1 + \sqrt{5}}{2}$, the golden ratio.
	By using and extending lower bounds known from the literature,
	we show that no deterministic online algorithm to minimize the makespan or total completion time in PFFBs with two stages can, in general, be better than the \mbox{t-Switch} algorithm.
	
	Both results are based on the observation that for all objectives we consider, there exists an optimal permutation schedule with start times of jobs ordered by a non-decreasing release date ordering on all stages. We show this as an extension to the theorem proved by \cite{PFBPaper1}. 

	Notice that for the offline version, our results imply that the Never-Wait algorithm is a 2-approximation algorithm for PFFBs to minimize the makespan, total completion time, maximum flow time or total flow time. This is interesting as so far, for PFFBs with arbitrarily many stages, no exact polynomial algorithm has been found even in the case where all stages consist of only one machine (see, e.g., \cite{PFBPaper1}).
	
	As this is the first study of online PFFBs, it is natural that some open questions remain. Most importantly, there remains a gap between the competitiveness of the Never-Wait algorithm and the lower bound of competitiveness in the case of three or more stages for makespan and total completion time, and even in the case of two stages, for the maximum flow time.
	Of course, it would be desirable to close this gap, either by proving a larger lower bound of competitiveness or by finding a better algorithm than the Never-Wait algorithm. As we have shown in Example \ref{Exa:NeverWaitEpsilon}, the Never-Wait algorithm itself cannot be better than 2-competitive in general.
	One way to improve the Never-Wait algorithm could be to better forecast what happens on the later stages of the PFFB. Observe that, despite the online situation, we can forecast job arrivals on later stages once the jobs become available at the first stage. Indeed, at any time step $t$, the full schedule (and thus arrival forecast) at each stage can be computed for all jobs which become available before $t$. Thus, for the later stages in the PFFB, the Never-Wait algorithm might be improved by using this additional knowledge of future job arrivals. Note that in the practical study mentioned above, such an improvement has been successfully attempted for a specific problem instance (see \cite{ackermannGOR}).
	
	In addition to closing these gaps, for the total completion time and total flow time objectives, it might be possible to achieve better competitive ratios for special cases where a certain minimum number of machines per stage is guaranteed. Observe that the lower bound constructions in Section~\ref{Sec:LowerBoundCmp} (Theorems~\ref{Thm:sumCjLB} and~\ref{Th:LBFSum}) only use a single machine. It might be the case that the same lower bounds do no longer hold if each stage contains several parallel machines. Possibly, better competitive ratios dependent on $m\coloneqq\min_{i\in[s]}m_i$ could be established. \citet{cao2011sumwjCj} provide a result of this kind for parallel machines (only one stage) with unbounded batch capacity. For the makespan and maximum flow time objectives, however, the lower bounds by \citet{zhang2003optimal} and \citet{Jiao2014maxflowtime} involve arbitrarily many parallel machines. Hence, for these objectives, it is not possible to achieve better competitive ratios for high values of $m$.
	
	Another open question concerns different objective functions. We have shown that in a schedule computed by the Never-Wait algorithm the finishing time of a job $J_j$ is at most twice the least possible finishing time of $J_j$ in any permutation schedule ordered by release dates. Unfortunately, for most traditional scheduling objectives beyond those studied in this paper, permutation schedules ordered by release dates are not, in general, optimal (see, e.g., \citet[Example~4]{PFBPaper1}).
	Still, the Never-Wait strategy may help with these objective functions, if jobs are prioritized differently. For example, instead of scheduling jobs in order of their release dates, each time a machine is started in the Never-Wait algorithm, one may instead pick the available jobs with the largest weight, if the objective involves job weights.
	It is, at the moment, unclear whether such an algorithm may be competitive and, if yes, what its competitiveness bound would be.
	
	On the other hand, in the online scenario, it may be valid to restrict the study of other objectives to the case where jobs are ordered by their release dates on all stages.
	In other words, instead of searching for an optimal schedule amongst all possible schedules, we search for an optimal schedule amongst all permutation schedules with jobs ordered by release dates.
	Especially in the pharmacological application we consider, such a first-in-first-out approach may be mandated due to ethical and fairness considerations.
	For this type of restricted problem, the Never-Wait algorithm may well prove to be competitive for many objectives beyond the ones considered in this paper, as long as these objectives are regular. For example, using the same arguments as before, the Never-Wait algorithm is 2-competitive w.r.t. the weighted total completion time / flow time. Indeed, as for the non-weighted versions, the factor 2 from Theorem \ref{NeverWaitBound} can be moved in front of the sum objectives to immediately see 2-competitiveness.

	\bibliographystyle{spbasic}      
	\bibliography{Online_PFFB_literature}   
	
	%
	%
	
	\section*{Appendix A: Proof of Theorem \ref{Thm:PermI}}
	We now show how to prove Theorem \ref{Thm:PermI}. First, analogous to Lemma~2 of \citet{PFBPaper1}, we need the following lemma.
	\begin{lemma}
		\label{lem:releaseDatePermu}
		Let $\varsigma$ be a feasible schedule for a PFFB and let $\pi$ be some earliest release date ordering of the jobs.
		Then there exists a feasible permutation schedule $\hat{\varsigma}$ in which the jobs are ordered by $\pi$ and the multi-set of job completion times in $\hat{\varsigma}$ is the same as in $\varsigma$.
	\end{lemma}
	\begin{proof}
		The proof is completely analogous to the proof of Lemma 2 from \cite{PFBPaper1}. To see this, one only needs to note that nowhere in the proof from \cite{PFBPaper1} it is actually needed that all batches $B_\ell^{(i)}$ at stage $S_i$ are processed on the same machine. Numbering the batches at stage $S_i$ in any start time order, the construction of the new schedule $\hat{\varsigma}$ as well as the proof of its feasibility work exactly as in the proof from \cite{PFBPaper1}. \qed
	\end{proof}

	Now we are ready to proof the main theorem.
	
	{\renewcommand*{\proofname}{Proof of Theorem \ref{Thm:PermI}}
	\begin{proof}
		Let $\varsigma$ be an optimal PFFB schedule with respect to one of our four objective functions $C_{\max}$, $\sum C_j$, $F_{\max}$, and $\sum F_j$. Let $\pi$ be an earliest release date ordering. Using Lemma \ref{lem:releaseDatePermu}, construct a new permutation schedule $\hat{\varsigma}$, with jobs ordered by $\pi$ on all stages and with the same multi-set of job completion times.
	
		For objective functions $C_{\max}$ and $\sum C_j$, clearly the new schedule $\hat{\varsigma}$ is optimal, since $\varsigma$ is optimal and $\hat{\varsigma}$ has the same multi-set of job completion times. Moreover, since $\sum F_j = \sum C_j - c$, where $c$ is a constant given by $c = \sum_{i=1}^n r_i$, the same argument holds for objective function $\sum F_j$.
		
		Finally, for objective function $F_{\max}$, suppose that jobs are indexed according to the earliest release date ordering $\pi$, i.e., $r_1\leq r_2\leq\dots\leq r_n$ and $C_1(\hat{\varsigma})\leq C_2(\hat{\varsigma})\leq\dots\leq C_n(\hat{\varsigma})$. Let $j$ be the index of the job with maximum flow time in $\hat{\varsigma}$, i.e., $F_{\max}(\hat{\varsigma})=C_j(\hat{\varsigma})-r_j$. Since $\varsigma$ has the same multi-set of job completion times as $\hat{\varsigma}$, there exist at most $j-1$ jobs with a completion time strictly less than $C_j(\hat{\varsigma})$ in the original schedule $\varsigma$. On the other hand, all the $j$ jobs $J_1,J_2,\dots,J_j$ have a release date of at most $r_j$. Hence, by pigeon-hole principle, there must exist a job $J_{j'}$ with $C_{j'}(\varsigma)\geq C_j(\hat{\varsigma})$ and $r_{j'}\leq r_j$. This implies \[F_{\max}(\varsigma)\geq C_{j'}(\varsigma) - r_{j'} \geq C_j(\hat{\varsigma})-r_j = F_{\max}(\hat{\varsigma}).\tag*{\qed}\]
	\end{proof} }
	
\end{document}